\documentclass[acmsmall,nonacm]{acmart}

\usepackage{hyperref}
\usepackage{cleveref}
\usepackage{tabularx}
\usepackage{multirow}
\usepackage{makecell}

\AtBeginDocument{%
  }

\setcopyright{acmcopyright}
\copyrightyear{2024}
\acmYear{2024}
\acmDOI{XXXXXXX.XXXXXXX}

\acmConference[PODC 2024]{ACM Symposium on Principles of Distributed Computing}{June 17--21,
  2024}{Nantes, France}

\author{Giuliano Losa}
\email{giuliano@stellar.org}
\orcid{0000-0003-2341-7928}
\affiliation{%
  \institution{Stellar Development Foundation}
  \country{USA}
}
\author{Eli Gafni}
\email{eli@ucla.edu}
\affiliation{%
  \institution{University of California, Los Angeles}
  \country{USA}
}
\affiliation{%
  \institution{Stellar Development Foundation}
  \country{USA}
}

\begin{document}

\title{Understanding Read-Write Wait-Free Coverings in the Fully-Anonymous Shared-Memory Model}

\begin{abstract}
    In the fully-anonymous (shared-memory) model, inspired by a biological setting, processors have no identifiers and memory locations are anonymous, meaning there is no pre-existing agreement among processors on any naming of the memory locations.
    In this work, we ask fundamental questions about the fully-anonymous model in the hope to obtain a better understanding of the role of naming and anonymity in distributed computing.

    First, we ask what it means to solve a task under processor anonymity.
    With tasks such as renaming, the traditional notion obviously does not apply. Instead of restricting ourselves to colorless tasks, we propose using the notion of group solvability, which allows transferring any task to processor-anonymous models.

    Second, the difficulty with anonymity is that processors can hardly avoid covering and then overwriting each other's writes, erasing information written by their predecessors.
    To get to the bottom of this phenomenon, we ask what system configurations are stable when processors keep reading and writing ad infinitum.
    Resolving this question leads us to a wait-free solution to the snapshot task, which then allows us to solve renaming and obstruction-free consensus.
\end{abstract}

\maketitle

\section{Introduction}

Taubenfeld and Raynal~\cite{raynal_mutual_2020} recently introduced the fully-anonymous (shared-memory) model, in which both the processors and the shared-memory are anonymous.
Processors anonymity means that all processors execute exactly the same program and start from the same initial state, except for their possibly different (but not necessarily so) external inputs.
Processors communicate using a number of shared multi-writer multi-reader (MWMR) atomic registers.
Crucially, the registers are also anonymous.
This means that each processor can only address the registers using a private local numbering of the registers, as there is no pre-existing agreement among the processors on their local numbering of the registers.

To make sense of this model, one can think of each processor being wired to the shared registers in an arbitrary way.
Or, inspired by Rashid et al.~\cite{rashidEpigeneticConsensusProblem2021,rashid_genome-wide_2020}, one can think of biological agents physically accessing locations in space without a common frame of reference.

Understanding what synchronization problems are solvable in full anonymity is a fundamental distributed-computing question, and yet not much is known about the fully-anonymous model.

An important difficulty with anonymous processors is that no processor can a priori reserve an exclusive location to write to, and thus processors might overwrite each other and erase important information.
Adding memory anonymity on top of processor anonymity compounds the problem.

For example, if only the processors are anonymous (but memory is not), Guerraoui and Ruppert~\cite{guerraoui_what_2005} show that it is possible to obtain wait-free atomic memory snapshots~\cite{afek_atomic_1993}, which they use to solve consensus with obstruction-free termination.
To keep processors from always overwriting each other, they have processors participate in a race, starting from a common initial position in a one-dimensional array, to be the first to write at a position in the array.
With anonymous memory, there is no way to even define a common starting register for the race or a shared ordering of the registers to race through, and this scheme does not work.

Since processors can hardly avoid overwriting each other's writes, it is not even clear that we can implement just writes that are guaranteed to persist, let alone atomic memory snapshots or consensus even if assuming obstruction-freedom.

To approach anonymity in a principled, tractable way, we start by focusing on the fundamental building blocks of distributed computing, namely \emph{tasks}\footnote{See Herlihy, Kozlov, and Rajsbaum~\cite[Section 8.2]{herlihyDistributedComputingCombinatorial2013} for a formal definition of tasks.}, and by asking what tasks are \emph{wait-free solvable} in the fully-anonymous model.
We expect to make our work easier by restricting ourselves to the well-behaved mathematics of wait-free task solvability instead of attacking ad hoc synchronization problems.

This focus on tasks immediately turns up an interesting observation, namely that the usual notion of task solvability is not adequate in processor-anonymous models.
To see why, consider the classic snapshot \emph{task}. %
In this task, each processor gets a unique identity as input and must return a set of participating input identities including its own and such that the sets returned by the participants are related by containment.
Obviously, in processor-anonymous models, we cannot assume that processors get unique identities.
Prior works~\cite{guerraoui_what_2005,raynal_fully_2021} sidestep this issue by resorting to atomic memory snapshots, i.e.\ snapshots of the contents of the memory at a point in time\footnote{Obtaining atomic snapshots of the memory is not the same as solving the snapshot task. The snapshot task is model-agnostic and does not depend on any notion of memory or registers; it is just about producing sets of participating identifiers related by containment. In a shared-memory implementation of the snapshot task, the outputs may very well not correspond to the contents of the memory at any point in time.}.

While an identity can be thought of as is an input, in processor-anonymous model we lack uniqueness.
This calls for considering groups consisting of the processors with the same input.
What is a reasonable task correctness condition when we have groups?
This question was first asked by Gafni~\cite{gafni_group-solvability_2004} and the resolution there, which we adopt here, is to view a task as referring to groups rather than individual processors and to require that, should we arbitrarily pick one representative processor from each group, the corresponding mapping from groups to outputs satisfy the task specification.
This gives us a definition of group-solving a task.

We now give some examples.
To group-solve the consensus task, it is easy to see that processors must return a unique participating group identifier.
To group-solve the snapshot task, each processor must return a set of group identifiers, including its group and consisting of participating groups, such that, if we pick one processor per group, the sets they return are related by containment.
Interestingly, this allows two processors belonging to the same group to return incomparable sets.
When group-solving renaming, processors in the same group are allowed to share a name, but two processors from different groups cannot.
We can similarly apply the definition to any other classic task, e.g.\ immediate-snapshot, set-consensus, weak symmetry breaking, etc.
See~\Cref{sec:solvability} for more precise definitions and additional examples.

With an adequate notion of task solvability in hand, we next present solutions to 3 tasks never solved before in the fully-anonymous model:
\begin{itemize}
    \item The snapshot task, with wait-free termination.
    \item The adaptive renaming task, with parameter $M(M+1)/2$ and wait-free termination, where $M$ is the number of groups.
    \item The consensus task, with obstruction-free termination.
\end{itemize}
Remarkably, each of those solutions uses only $N$ registers, where $N$ is the number of processors.

The major algorithmic contribution is the solution to the snapshot task, requiring an entirely novel construction.
Once we obtain the snapshot task, the other two follow using algorithms presented in prior work~\cite{chandra_polylog_1996,bar-noy_shared-memory_1989}.
Nevertheless, proving that those algorithms are still correct when substituting a group-solution to the snapshot task --- which, as we have observed, allows processors in the same group to return unrelated snapshots --- instead of atomic memory snapshots, is subtle. %

Finally, the key intuition used to group-solve the snapshot task came by formulating and solving the ``eventual pattern'' problem, which we find is of independent interest.
The question is to characterize the structure formed by set of views can that can be maintained forever, called stable views, when executing a simple write-scan loop forever.
We tackle this question in~\Cref{sec:stable-views} and show that stable views must form a directed acyclic graph with a single source.
This in turn suggests the main idea behind the snapshot algorithm of~\Cref{sec:snapshot}.

\section{The Fully-Anonymous Model}

In this section we describe the fully-anonymous model, which we use in the rest of the paper.
We roughly follow Raynal and Taubenfeld~\cite{raynal_fully_2021}.

We consider a set of $N>1$ processors each identified by a unique number in the range $1..N$.
Each processor runs a deterministic, sequential program, and processors take steps asynchronously, i.e.\ independently of each other.

Processors know $N$ but are otherwise anonymous, meaning that they all run exactly the same program.
In particular, the identifier of a processor does not appear in its program.

Processors communicate using $M>0$ shared multi-reader, multi-writer atomic registers $register[1]$, $register[2]$, \dots, $register[M]$ by issuing read and write instructions that atomically read or write a single register at a time.
However, while processors know there are $M$ registers, there is no pre-existing agreement among the processors about any numbering of the registers.
Instead, for each processor $p$, there is a permutation~$\sigma_p$ of $1..M$, unknown to the processors (including $p$) and fixed arbitrarily at initialization, such that a read or write  instructions by processor $p$ of register number $i$ reads or writes, respectively, register $register[\sigma_p[i]]$.
All registers initially contain a known default value.

When discussing executions, we say that a processor $p$ reads register $i$ when it reads the contents of the register $register[i]$, i.e.\ $p$ executes an instruction to read the register $\sigma_p^{-1}[i]$.
Similarly, we say that a processor $p$ writes register $i$ when it writes the contents of the register $register[i]$, i.e.\ $p$ executes an instruction to write the register $\sigma_p^{-1}[i]$.

The state of each processor consists of three components: a private, read-only input; a private, mutable local state; and a private, write-once output.
A state of the system determines, for each processor, the permutation describing its wiring to the registers (which does not change) and its state, and the contents of each register.

A step is an ordered pair of states, called pre-state and post-state, that is either a read step, a write step, a local computation step, or an output step of a single processor, each according to its program.

An execution $e$ is an infinite sequence of states $\langle e[0], e[1], e[2], \dots\rangle$ such that:
\begin{enumerate}
    \item in the initial state $e[0]$, processor's permutations and inputs are arbitrary, and every processor is in the same, designated initial local state and has an empty output, and all registers are empty;
    \item for each time $t\in Nat$, the ordered pair $\langle e[t],e[t+1]\rangle$ is a step of a processor.
\end{enumerate}
We say that a processor participates in an execution when it takes at least one step.

For every time point $t\in Nat$, we say that $s[t]$ is the state at time $t$, and for every time point $t'\in Nat$ where $t'\geq t$, we say that the state $s[t']$ is reachable after time $t$.
The step occurring at time $t$ is the ordered pair $\langle e[t],e[t+1]\rangle$, and we say that a step $stp$ occurs after time $t$ when there exists $t'\geq t$ such that $stp=\langle e[t'],e[t'+1]\rangle$.

We say that a processor $p$ reads from a processor $q$ at time $t$ if the step occurring at time $t$ is a read step of a register $r[i]$ by $p$ and, at time $t$, register $r[i]$ was last written by processor $q$.
We also say that a processor $p$ reads from a set of processors $Q$ at time $t$ when $p$ reads from a member of $Q$ at time~$t$.

\subsection{Read-Write Wait-Free Synchronization is Impossible with Fewer Than $N$ Registers}
\label{sec:no-less-than-n}

The main difficulty in the fully-anonymous model is that, because processors are wired to the registers arbitrarily, it is hard to avoid processors overwriting each other's writes.
The following easy lower bound illustrates this phenomenon.
This result has been proven for stronger models~\cite{delporte-gallet_linear_2015} and as such subsumes the fully anonymous model, yet it is such a simple observation for the fully anonymous model that we bring it to attention here.

Consider the fully-anonymous model with $N-1$ registers.
Pick a processor $p$ and let $Q=P\setminus \{p\}$.

Let the members of $Q$ run, while $p$ does not take steps, until all members of $Q$ are poised to perform their first write.
Moreover, suppose the wiring of processors to registers is such that each member of $Q$ is poised to write to a different register.
Now let $p$ run solo until it produces an output.
Finally let all the member of $Q$ write, and notice that no information written by $p$ remains in the system.

To the member of $Q$, this execution is indistinguishable from another execution in which $p$ has a different input.
Vice versa, to $p$, this execution is indistinguishable from another execution in which the members of $Q$ have different inputs.
Thus no read-write coordination is possible between $p$ and $Q$.

\section{Defining Task Solvability in the Fully-Anonymous Model}
\label{sec:solvability}

\subsection{Tasks}

Tasks are a class of distributed computing problems that are particularly mathematically well-behaved;
 see for example~\cite{herlihy_distributed_2013}.

In this paper, we will only be interested in tasks where every processor receives its own identifier as input.
Thus, we specify a task $\mathcal{T}$ as a set of outputs $\mathcal{O}$ and a set of output assignments $\Delta$, where an output assignment is a partial function from processors to outputs.

Given an execution in which all participating processors terminate, we define the executions' output assignment in the obvious way.
Then, we say that an algorithm solves a task $T$ when, in every execution $e$ in which all participating processors terminate, the output assignment is in $\Delta$.

We will be particularly interested in the following classic tasks: consensus, snapshot, and adaptive renaming.
\begin{definition}[Consensus]
        In the consensus task, all processors must agree on the identifier of a participating processor.
        Formally, the set of outputs is the set of processor identifiers and the set of valid output assignments is the set of constant partial functions on processor identifiers mapping to their domain of definition.
\end{definition}
\begin{definition}[Snapshot]
        In the snapshot task, each processor $p_i$ outputs a set of participating processor identifiers $o[i]$ such that $p_i\in o[i]$ and such that every pair of outputs are related by containment (for every $i,j$ in the domain of $o$, either $o[i]\subseteq o[j]$ or $o[i]\subseteq o[j]$).
\end{definition}

\begin{definition}[Adaptive Renaming]
        In the adaptive renaming task with parameter $f$, where $f$ is a function on natural numbers, each processor $p_i$ outputs a unique natural number $o[i]$ such that, if $n$ processors participate, the outputs must be in the range $1..f[n]$.
\end{definition}
We will later be interested in renaming with parameter $f[n]=n(n+1)/2$.

\subsection{Group Solvability}
\label{sec:group}

In our formulation of the consensus task, processors must output the identity of a unique processor.
What will the consensus task be in the fully-anonymous model if processors cannot reveal their private identity?
Naturally we can think of giving each processor an input and asking them to agree on a unique input of a participating processor.
We can also think of this as giving processors non-unique identities, i.e.\ assigning them to groups, and then requiring the processors to agree on the identity of a participating group.

Adopting the group view, we could stipulate the snapshot task as follows: each processor $p$ belonging to group $g(p)$ must output a set $S_p$ of (participating) group identifiers, such that $g(p) \in S_p$ and for every $p,q$, $S_p$ and $S_q$ are related by containment (i.e.\ either $S_p\subseteq S_q$ or vice versa).

This informal construction of group versions of tasks cannot be easily generalized to say, the task of renaming.
For renaming, one idea would be to require processors in the same group to output the same name, but that would likely be asking them to solve consensus.
Can we have an interpretation that will encompass both snapshot and renaming?

Gafni~\cite{gafni_group-solvability_2004} proposes an answer which we adopt in this paper.
Gafni's idea is to view a task as referring to groups rather than individual processors, and to require that, should we consider an execution and arbitrarily pick one representative processor from each group, the corresponding mapping from groups to outputs must satisfy the task specification.
This gives us a definition of group-solving any task.
A formal definition appears in~\Cref{def:group-solving} of~\Cref{sec:group-solvabilty-def}.

In the case of renaming, note that this resolves our conundrum above by allowing (but not requiring) processors in the same group to output the same name.

Surprisingly, in the case of the snapshot task, Gafni's definition does not coincide with our first attempt above.
For example:
consider 4 processors $1,2,3,4$ and let the groups be $A= \{1\}$, $B= \{2,3\}$, and $C= \{4\}$.
Let processor 1 output \{A,B,C\}, processor 2 output \{A,B\}, processor 3 output \{B,C\}, and processor 4 output \{A,B,C\}.
This is a legal group solution to the snapshot task according to Gafni's definition, even though the two processors in group $B$ (processors 2 and 3) return incomparable sets.

Previous authors dealing with processor-anonymity missed the group connection, maybe because they thought only about atomic memory snapshots and not the snapshot task.

Finally, let us mention that the snapshot algorithm that we present in~\Cref{sec:snapshot} unintentionally achieves the stronger requirement of ensuring that all outputs are related by containment (as in our initial attempt above).
Nevertheless, we think that Gafni's group solvability is the right notion to use because it allows transferring all tasks, including non-colorless ones, to the processor-anonymous models.
It does however complicate the correctness proofs of the renaming algorithm of \Cref{sec:renaming}.
In fact, it is surprising that this algorithm, borrowed from previous work~\cite{bar-noy_shared-memory_1989}, works at all given a group solution to the snapshot task instead of atomic memory snapshots.

\subsubsection{Group Solvability: Formal Definition}
\label{sec:group-solvabilty-def}

Consider a task $\mathcal{T}=\langle \mathcal{O}, \Delta\rangle$ for a number $N_\mathcal{T}$ of processors, with $N_\mathcal{T}$ possibly different from $N$ the number of processors in the system.

Consider an execution $e$ in which processors receive inputs in the range $1..N_{\mathcal{T}}$.
Define the groups $G_1,\dots,G_{N_\mathcal{T}}$ where $G_i$ is the set of all processors with input $i$.
Note that there are as many groups as there are processors in the definition of $\mathcal{T}$.

We say that $i$ is the identifier of group $G_i$, and we say that a group participates in the execution $e$ if at least one of its members participates in~$e$.

Now define an output sample of $e$ as any function mapping the identifier of each participating group to the output of one of its members.

\begin{definition}
    \label{def:group-solving}
    An algorithm group-solves the task $\mathcal{T}$ when, for every execution $e$ and every output sample $o$, $o$ is a valid output of the task.
\end{definition}

\section{Warm-up: The Eventual Pattern and the Structure of Stable Views}
\label{sec:stable-views}

Towards solving the snapshot task, in this section we consider a scenario in which each processor gets an arbitrary input taken from a set of values $V$ and then indefinitely alternates between writing the set of inputs it knows about, which we call its view, in a shared register, and then reading all the registers one by one, adding the contents of the registers to its view.
We assume that each processor issues writes fairly, i.e.\ it writes all the registers once before writing any register for the second time, etc. %
Pseudocode for the algorithm appears in~\Cref{fig:write-scan}.

\begin{figure}[h]
\centering
\includegraphics[]{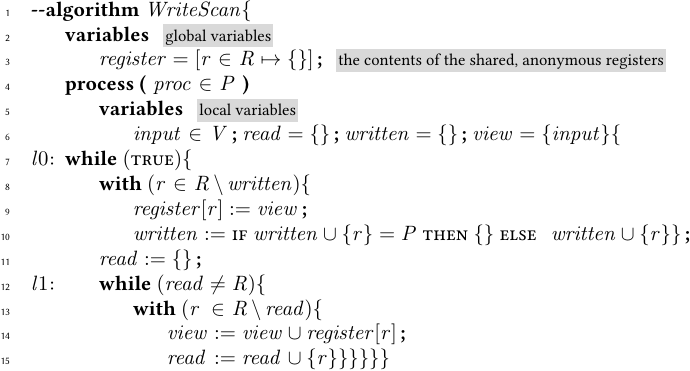}
\caption{The write-scan algorithm in the PlusCal algorithm language. Note that the sequence of steps between any two labels (here the labels are $l0$, line 7, and $l1$, line 12) is executed atomically.}
\label{fig:write-scan}
\end{figure}

When can a processor terminate and declare its view a snapshot?
At first sight, one might say that a processor can terminate when it reads the same set of values in every register.
Unfortunately, this does not work.
Then maybe a double collect will work, i.e.\ reading the same set of values in every register twice in a row?
Neither does this work.

In fact, as we show below, it is possible for a processor $p$ to read the same set of values $V$ in all registers ad infinitum and for another processor $p'$ to also read the same set of values $V'$ in all registers ad infinitum, and yet neither set is a subset of the other.
Moreover, this can happen regardless of the number of registers.

\subsection{A pathological infinite execution}
\label{sec:infinite-exec}

Concretely, first consider 3 processors $p_1$, $p_2$, and $p_3$, each receiving inputs $1$, $2$, and $3$, respectively, in a system with 3 registers, and consider the infinite execution depicted in~\Cref{fig:infinite}.
In this execution, the processor's wiring to the registers (i.e.\ the processors' permutations) are such that the processors keep overwriting each other's writes, as depicted.
Thus, despite taking infinitely many steps, processors $p_2$ and $p_3$'s views --- $\{1,3\}$ and $\{1,2\}$, respectively --- remain incomparable.

Because $p_2$ and $p_3$ keep scanning different sets of values, this does not yet show that e.g.\ requiring a processor to see the same set of values in all registers before outputting a snapshot would not work.
However, now introduce two additional processors $p$ and $p'$, both with input $1$.
Let $p$ read $\{1,2\}$ each time $p_2$ writes it, and when $p$ has to write, let $p$ write $\{1,2\}$ immediately after $p_2$ writes it and to the same register.
Thus all $p$ ever sees is $\{1,2\}$, and $p$ never perturbs the execution of $p_1$, $p_2$, or $p_3$.
Similarly, we can arrange for $p'$ to only ever see $\{1,3\}$ and otherwise not perturb the execution.
We now have our full example: $p$ can take infinitely many steps and only every see $\{1,2\}$ in all registers, while $p'$ can take infinitely many steps and only every see $\{1,3\}$ in all registers.

Additionally, note that adding one more register $r_4$ would not help prevent this type of execution; it would merely add three more overwriting steps to complete the repeating cycle.
Similarly, no additional number of registers would prevent this type of infinite execution.

\begin{figure}
\centering
\begin{tabularx}{\textwidth} {
  | >{\centering}l
  | >{\centering}l
  | >{\centering\arraybackslash}X
  | >{\centering\arraybackslash}X
  | >{\centering\arraybackslash}X
  | >{\centering\arraybackslash}X
  | >{\centering\arraybackslash}X
  | >{\centering\arraybackslash}X | }
 \hline
 & \multirow{2}{2em}{Actions} & \multicolumn{6}{c|}{Post State} \\
 \cline{3-8}
 &                            & $r_1$ & $r_2$ & $r_3$ & $view[p_1]$ & $view[p_2]$ & $view[p_3]$ \\
                           \hline
 1 & \makecell[cl]{$p_1$ writes twice and\\ ends with a scan}
                           & $\{\}$  & $\{1\}$  & $\{1\}$ & $\{1\}$ & $\{2\}$ & $\{3\}$ \\
                           \hline
 2 & $p_2$ writes then scans
                           & \underline{$\{2\}$}  & $\{1\}$  & $\{1\}$ & $\{1\}$ & \underline{$\{1,2\}$} & $\{3\}$ \\
                           \hline
 3 & \makecell[cl]{$p_3$ overwrites $p_2$\\ then scans}
                           & \underline{$\{3\}$}  & $\{1\}$  & $\{1\}$ & $\{1\}$ & $\{1,2\}$ & \underline{$\{1,3\}$} \\
                           \hline
 4 & \makecell[cl]{$p_1$ overwrites $p_3$\\ then scans}
                           & \underline{$\{1\}$}  & $\{1\}$  & $\{1\}$ & $\{1\}$ & $\{1,2\}$ & $\{1,3\}$ \\
                           \hline
 5 & $p_2$ writes then scans
                           & $\{1\}$  & \underline{$\{1,2\}$}  & $\{1\}$ & $\{1\}$ & $\{1,2\}$ & $\{1,3\}$ \\
                           \hline
 6 & \makecell[cl]{$p_3$ overwrites $p_2$\\ then scans}
                           & $\{1\}$  & \underline{$\{1,3\}$}  & $\{1\}$  & $\{1\}$ & $\{1,2\}$ & $\{1,3\}$ \\
                           \hline
 7 & \makecell[cl]{$p_1$ overwrites $p_3$\\ then scans}
                           & $\{1\}$  & \underline{$\{1\}$} & $\{1\}$  & $\{1\}$ & $\{1,2\}$ & $\{1,3\}$ \\
                           \hline
 8 & $p_2$ writes then scans
                           & $\{1\}$  & $\{1\}$  & \underline{$\{1,2\}$} & $\{1\}$ & $\{1,2\}$ & $\{1,3\}$ \\
                           \hline
 9 & \makecell[cl]{$p_3$ overwrites $p_2$\\ then scans}
                           & $\{1\}$  & $\{1\}$  & \underline{$\{1,3\}$} & $\{1\}$ & $\{1,2\}$ & $\{1,3\}$ \\
                           \hline
 10 & \makecell[cl]{$p_1$ overwrites $p_3$\\ then scans}
                           & $\{1\}$  & $\{1\}$  & \underline{$\{1\}$} & $\{1\}$ & $\{1,2\}$ & $\{1,3\}$ \\
                           \hline
 11 & $p_2$ writes then scans
                           & \underline{$\{1,2\}$}  & $\{1\}$  & $\{1\}$ & $\{1\}$ & $\{1,2\}$ & $\{1,3\}$ \\
                           \hline
 12 & \makecell[cl]{$p_3$ overwrites $p_2$\\ then scans}
                           & \underline{$\{1,3\}$}  & $\{1\}$  & $\{1\}$ & $\{1\}$ & $\{1,2\}$ & $\{1,3\}$ \\
                           \hline
 13 & \makecell[cl]{$p_1$ overwrites $p_3$\\ then scans (same as 4)} %
                           & \underline{$\{1\}$}  & $\{1\}$  & $\{1\}$ & $\{1\}$ & $\{1,2\}$ & $\{1,3\}$ \\
                           \hline
\end{tabularx}
\caption{Infinite execution where $p_2$ and $p_3$ keep writing incomparable sets of values.
    Steps 5 to 13 repeat forever after step 13. Each step, changes in the post-state are underlined.}
\label{fig:infinite}
\end{figure}

\subsection{Stable views and the global stabilization time GST}

Obviously, the view of a processor can only grow as it takes more steps.
Thus, eventually, we must get to a point after which no view ever changes anymore.

Say that a processor is live if it takes infinitely many steps, and that a set of processors is live when at least one of its members is live.

\begin{definition}[The global stabilization time, GST]
    Let GST be the earliest time after which all views are stable, all processors that are not live have taken their last step, and all writes by non-live processors have been overwritten by live processors.
\end{definition}

\begin{definition}[Stable view]
    We say that the view of a live processor after GST is a stable view.
\end{definition}
Note that, if a processor just stops taking steps, we do not count its view as stable.

What can we say about stable views?
This question turns out to uncover an interesting structural property of the fully-anonymous model.
Namely, as we show in the next section, that the stable views in an infinite execution always form a directed acyclic graph with a single source.

In~\Cref{sec:snapshot}, we then present a wait-free snapshot algorithm motivated by this single-source DAG structure of stable views\footnote{Hint: processors having the source view can terminate, then a new source will appear, and processors with that new source view can terminate, etc. But the trick will be to find how processors can detect that they have the source view.}.

\subsection{The Eventual Pattern: Stable views form a DAG with a unique source}
\label{sec:dag}

We now consider an infinite execution of the write-scan loop shown in~\Cref{fig:write-scan}. %

\begin{definition} [Stable-View Graph]
    The stable-view graph $G$ is the graph such that:
    \begin{enumerate}
        \item The vertices are the stable views appearing in the execution.
        \item There is an edge from a view $V_1$ to a view $V_2$ when $V_1\subset V_2$.
    \end{enumerate}
\end{definition}
In the example execution of~\Cref{fig:infinite}, the vertices of the stable-view graph are $\{1\}$, $\{1,2\}$, and $\{1,3\}$ and the edges are $\langle \{1\},\{1,2\}\rangle$ and $\langle \{1\},\{1,3\}\rangle$.

In the eventual-pattern question, we ask whether the stable-view graph has any characteristic structure.

Note that, by property of the subset relation $\subset$ (it is transitive and irreflexive), $G$ is obviously a directed acyclic graph.
Additionally, we now show that the stable-view graph has a unique source.
In the example of~\Cref{fig:infinite}, this unique source is the stable view $\{1\}$.

\begin{lemma}
    \label{lem:graph-1}
    For every two stable views $V_1$ and $V_2$, if $V_1\not\subset V_2$ then no live processor that has stable view $V_2$ reads from any processor that has view $V_1$.
\end{lemma}
\begin{proof}
    Suppose towards a contradiction that $V_1\not\subset V_2$ and a live processor $p$ with stable view $V_2$ reads at a time $t$ from a processor that has view $V_1$.
    Then, by the algorithm, $p$'s view at time $t+1$ is a superset of $V_1$.
    Moreover, since views only grow as processors take steps, we must have $V_1\subseteq V_2$, which is a contradiction.
\end{proof}

\begin{lemma}
    \label{lem:graph-2}
    Consider a nonempty set of processors $A$ that is live and a time $t$ such that, after time $t$, every read of every member of $A$ reads a register last written by a member of $A$.
    Then, at all times after time $t$, the set of registers last written by $\overline{A}$ is of cardinality at most equal to the cardinality of $A$.
\end{lemma}
\begin{proof}
    Suppose towards a contradiction that there is a time $t_2\geq t$ such that, at time $t_2$, the set of registers last written by members of $\overline{A}$ is of cardinality at least $|A|+1$.
    Consider the first processor $p$ in $A$ that completes a full scan that starts at a time $t_3\geq t_2$ and let $t_4$ be the time of $p$'s last read in this scan.
    Note that, since $A$ contains at leat one live processor, there is such a processor $p$.

    Because each processor writes exactly once in between each scan, there are at most $|A|$ writes by members of $A$ between times $t_2$ and $t_4$.
    However, we have assumed that, at time $t_2$, there are at least $|A|+1$ registers last written by members of $\overline{A}$.
    Thus, there is at least one register $r$ such that, at every time between $t_2$ and $t_4$, $r$ was last written by a member of $\overline{A}$.

    Since $p$ reads all the registers between times $t_3$ and $t_4$ and $t\leq t_2\leq t_3$, we conclude that $p$ must read $r$, a register last written by $\overline{A}$, at a time $t'\geq t$.
    This contradicts our assumptions.

\end{proof}

\begin{lemma}
    \label{lem:graph-3}
    If, at all times after a time $t$, there are at most $|A|$ registers last written by $\overline{A}$ and $\overline{A}$ is live, then at least one member of $\overline{A}$ eventually reads from $A$.
\end{lemma}
\begin{proof}
    Consider the set $R^{\overline{A}}_t$ of registers that, at time $t$, were last written by $\overline{A}$.

    First, suppose that $|R^{\overline{A}}_t|< |A|$.
    Then, the first time a member $p$ of $\overline{A}$ completes a scan that started after $t$, members of $\overline{A}$ wrote at most $|\overline{A}|$ times.
    This means that, at all times during $p$'s scan, at most $|R^{\overline{A}}_t|+|\overline{A}|< N$ registers were last written by members of $\overline{A}$.
    Thus, during its scan, $p$ must read a register last written by $A$, and we are done.

    Second, suppose that $|R^{\overline{A}}_t| = |A|$.
    If the set of registers last written by $\overline{A}$ never changes after $t$, then, trivially, there is a register $r$ such that, at all times after $t$, $r$ is last written by a member of $A$.
    Since all members of $\overline{A}$ eventually read all registers, we are done.

    Now suppose that the set of registers last written by $\overline{A}$ changes at some time $t'> t$, and let $R_{t'}$ be the set of registers that, at time $t'$, were last written by $\overline{A}$.
    Since we have assumed that there are at most $|A|$ registers last written by $\overline{A}$, the set of registers last written by $\overline{A}$ cannot grow, and thus it must shrink.
    Thus we have $|R_{t'}|< |A|$ and we can reason as in the first case above and we are done.
\end{proof}

Taken together, \Cref{lem:graph-2,lem:graph-3} show that it is not possible to have two disjoint sets of processors $A$ and $B$ that forever take steps and yet no member of $A$ ever reads from a member of $B$ and vice versa.
Formally, this is expressed in the following lemma, which follows directly from~\Cref{lem:graph-2,lem:graph-3}:
\begin{lemma}
    \label{lem:graph-key}
    Consider a nonempty set processors $A$ such that $A$ is live and a time $t$ such that, after time $t$, every read of every member of $A$ reads a register last written by a member of $A$.
    Assume that at least one member of $\overline{A}$ is live.
    Then, at least one member of $\overline{A}$ eventually reads from $A$. %
\end{lemma}

\begin{theorem}
    \label{thm:single-source}
    The stable-view graph is a directed acyclic graph with a unique source.
\end{theorem}
\begin{proof}
    Remember that, intuitively, \Cref{lem:graph-key} says that there cannot be two disjoint, live sets of processors where neither reads from the other.
    Now observe that, if the stable-view graph has two sources, then we can easily find two such disjoint sets, which is a contradiction.
    We now make a more formal argument.

    Suppose towards a contradiction that there are at least two sources in the graph, and consider a source $S$.
    Consider the nonempty set $X$ of views that are reachable from $S$ in the graph, and its nonempty complement among the graph's vertices $\overline{X}$.
    Let $P_X$ be the non-empty set of live processors whose stable view is in $X$.
    By the contrapositive of~\Cref{lem:graph-1}, after GST, for every processor $p\in P_X$, every read of $p$ reads from $P_X$. %
    Thus, by~\Cref{lem:graph-key}, there must a member $p$ of $\overline{X}$ that reads from $X$ at some time after GST.
    Thus, $p$'s stable view must be a superset of a view in $X$.
    This implies that a view in $\overline{X}$ is reachable from $X$, which in turns implies it is reachable from the source $S$, which by definition of~$\overline{X}$ is a contradiction.
\end{proof}

\section{A Solution to the Snapshot Task}
\label{sec:snapshot}

\subsection{From stable views to the snapshot algorithm}
\label{sec:dag-to-snapshot}

In~\Cref{sec:dag}, we have seen that the stable views in an infinite execution of the write-scan loop in~\Cref{fig:snapshot} form a directed acyclic graph with a single source.
The rough idea behind the snapshot algorithm that we present in this section is to let processors that have the source view output the view as snapshot, knowing that all other processors have a larger view.
This will then give rise to a new stable-view graph, and the new source can output, etc. until all processors have terminated.

But how does a processor detect that it has a view that can be maintained forever, and that, should the view be maintained forever, it would be the source of the stable-view graph?
Moreover, what if other processors have not reached their stable views yet?

In short, the stable-view graph describes infinite executions, but to solve snapshot wait-free we care only about finite executions.
So, while the results of~\Cref{sec:stable-views} allowed us to build some intuition, they are not sufficient to obtain a wait-free snapshot algorithm.

Let us turn again to~\Cref{sec:stable-views} but for a different intuition.
Consider the infinite execution with 5 processors described in~\Cref{sec:infinite-exec} and partially depicted in~\Cref{fig:infinite}.
Note that $p$ and $p'$, which keep reading the same incomparable sets ad infinitum, are able to do so because of the churning activity of $p_1$, $p_2$, and $p_3$ which produce the incomparable sets read by $p$ and $p$'. %

How can we prevent processors $p_1$, $p_2$, and $p_3$ from churning out incomparable sets ad infinitum?
Notice that, in~\Cref{fig:infinite}, $p_2$ and $p_3$ never complete a scan reading the same set in all registers.
So, if $p_2$ and $p_3$ would write something along with their view to warn $p$ and $p'$ of this fact, then $p$ and $p'$ could detect that something is wrong and not terminate yet.
This observation gives us a toehold on the solution.

More generally, here is the idea to prevent groups of processors from churning out incomparable sets as in~\Cref{fig:infinite}: %
Each processor $p$ assigns itself a level, initially $0$, and, each scan, $p$ sets its level to $0$ if it reads a set different from its view, and otherwise sets its level to the minimum level read in the scan plus 1.
Finally, a processor outputs a snapshot only when reaching level $N$.

Now consider a processor $p$ that outputs a set $W$.
If we ask who $p$ read $W$ from at level $N-1$, and in turn who those read $W$ from at level $N-2$, and in turn who those read $W$ from at level $N-3$, etc. until level $0$, it is not possible to add a new processor at each level and still have processors left when reaching level $0$.
Thus we obtain a set of processors $Q$ with $p\in Q$ and such that each member of $Q$ read $W$ everywhere from members of $Q$.
This prevents the situation in~\Cref{fig:infinite}, where, to keep churning incomparable sets, $p_2$ and $p_3$ cannot read the same set in all registers, and $p$ and $p'$ must read from them.

Finally, the fact that stable views form a directed acyclic graph allow us to establish termination, as the processors with a source view must inevitably reach level $N$ and terminate.
In the example of~\Cref{fig:infinite}, $p_1$ always reads $\{1\}$ --- the source stable view --- from itself; thus, if it tracked its level as above, $p_1$ would increase its level at each read and eventually terminate returning snapshot $\{1\}$; this would break the infinitely repeating pattern.

This is so far just an informal argument, but we hope it helps the reader build some intuition.
In the next section, we present the snapshot algorithm precisely and provide a rigorous proof.

\subsection{The snapshot algorithm}

In this subsection, we describe the snapshot algorithm, a wait-free implementation of the snapshot task in the fully-anonymous model.
\Cref{fig:snapshot} gives a specification in the PlusCal language~\cite{lamport_pluscal_2009}.

Processors communicate using $N$ shared registers, each initially a record with two components: an empty view $register[r].view$ and a level $register[r].level$ of 0 (line 4).
To simplify the notation, if $W$ is a set of values, we write $W\subseteq r$ for $W\subseteq register[r].view$.

Each processor $p$ has an input and maintains a set of values, called its view and initialized to the singleton containing $p$'s own input value, and a level between $0$ and $N$, initialized to $0$ (lines 7 to 9; the other variables are implementation details).
Looking ahead, we will see that $p$'s view contains the input values of other processors that $p$ knows of (and thus it never decreases).

After initialization, $p$ alternates between a writing phase and a scanning phase.
In the writing phase (lines 12 to 16) $p$ picks a register that it has not written to since it last wrote all the registers and writes its views and its level in the register.

In the scanning phase (lines 17 to 24), $p$ first reads all the registers one by one (lines 17 to 19).
Then, if $p$ read its own view in all registers, then $p$ sets its level to the minimum level read in the registers plus one; otherwise, $p$ resets its level to 0.
Finally, $p$ adds all the values it read to its view.

Processor $p$ repeats the write-scan loop until it reaches level $N$, at which point it terminates and outputs its view as snapshot\footnote{In fact, reaching level $N-1$ is sufficient, but the correctness proof is easier with level $N$.}.

\begin{figure}
\centering
\includegraphics[]{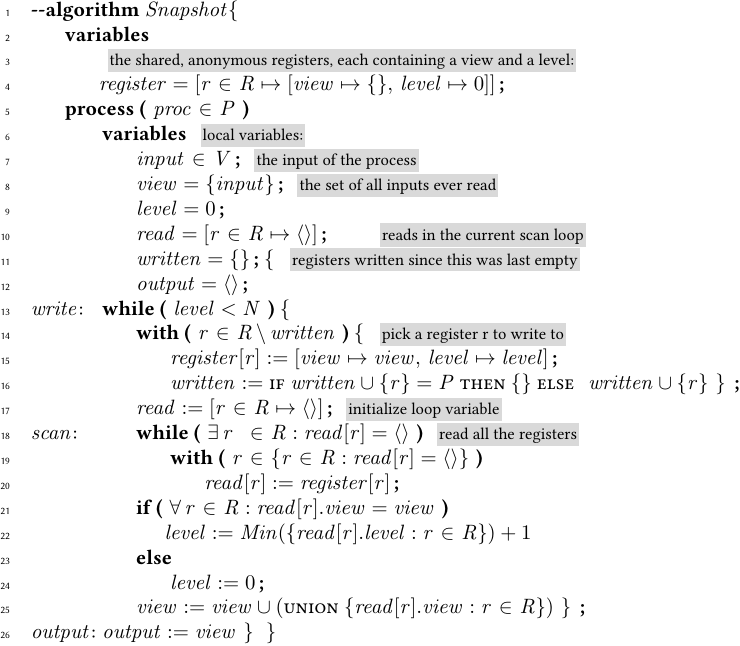}
\caption{The snapshot algorithm in the PlusCal~\cite{lamport_pluscal_2009} algorithm language.
    Each processor outputs its view when it terminates upon reaching level $N$.
    Note that the sequence of steps between any two labels (here the labels are $write$, line 12, and $scan$, line 17) is executed atomically.
    The TLC model-checker~\cite{lamport_specifying_2002} is able to exhaustively explore all 3-processor executions of this algorithm, and it confirms that the algorithm solves the snapshot task wait-free.}
\label{fig:snapshot}
\end{figure}

\subsection{Correctness}

\subsubsection{An informal operational argument of safety}

Let $p$ be the first processor to scan and obtain level $N$ with view $W$.
We inductively build a tree of processors layer by layer.
At the root of the tree is $p$, and edges to the next layer containing all the processors it
read from with level $N-1$, and continue inductively until the next layer does not strictly increase the 
cardinality of the set of processors residing in the tree.
Let this set be called $S_{p,N}$.

We now look at the instant just before the step $p$ does its first read in its last scan.
Consider processors in $S_{p,N}$.
If it contains a processor that only reads $W$ while $p$ does its scan, we repeat the argument replacing $p$ with the first $q$ that scans and obtains level $N-1$.
So now we assume an $i$ such that all of $S_{p,i}$ have $W$ when $p$ starts its scan.
The claim is that at the configuration before $p$'s first read there are at most $|\overline{S_{p,i}}|-1$ registers with values $W'\not\subseteq W$.
In this case, since $S_{p,i}$ is a lower bound to the number of processors that have $W$, it is easy to see that the rest of the processors with value different than $p$ cannot overwrite all of the registers containing $p$ value.

To see the claim is true: if the condition is violated and registers with $W'\not\subseteq W$ are numerous enough that $W$ can be evicted from all registers, then one of the processors in $S_{p,i} \setminus \{p\}$ must overwrite two registers while $p$ is scanning.
If we take the first to write twice, within $p$'s scan, it must read an non-$p$ value register --- contradiction.

\subsubsection{A proof of safety by induction}

We now rigorously prove that the algorithm of~\Cref{fig:snapshot} guarantees that every two output views are related by containment.
Note that this guarantee is stronger than what group-solving the snapshot task requires because, by the group-solvability definition, processors in the same group (i.e.\ with the same input) are allowed to return sets not related by containment.
We start by defining a notion that is central to the proof:
\begin{definition}[Set of values durably stored despite interference by a set of processors]
    Consider the system at a time $t$, a set of values $W$, a set of processors $Q$, and the following derived sets:
    \begin{itemize}
        \item $R_W\subseteq R$, the set of registers $r$ such that $W\subseteq r$.
        \item $Q_W\subseteq Q$, the set of processors $q\in Q$ such that either $W$ is a subset of $q$'s view or $q$ is scanning and has not yet read any register in $R_W$.
    \end{itemize}
    We say that, at time $t$, $W$ is durably stored despite interference by $Q$ when $|R_W| > |Q\setminus Q_W|$.
\end{definition}

Intuitively, if $W$ is durably stored despite interference by $Q$ and only members of $Q$ write, then only the members of $Q\setminus Q_W$ can overwrite the registers containing $W$ before they have to scan.
However, the condition $|R_W| > |Q\setminus Q_W|$ implies that the members of $Q\setminus Q_W$ are not numerous enough to overwrite all of $R_W$ before they have to scan.
Thus, if only the members of $Q$ write, then at least one register will always contain $W$.
In particular, taking $Q=P$ we have:
\begin{lemma}
    If, at a time $t$, $W$ is durably stored despite interference by $P$, then every processor $p$ that takes a step after $t$ and then terminates, terminates with $W\subseteq view[p]$.
\end{lemma}

Thus, to show that the algorithm is a safe snapshot implementation, we are going to show that:
\begin{lemma}
    \label{lem:durably-stored}
    If a processor $p$ terminates, then $view[p]$ is durably stored despite interference by $P$.
\end{lemma}

To prove~\Cref{lem:durably-stored}, it will be useful to consider the notion of a write covering a register during a scan of a process $p$:
\begin{definition}[Write covering a register during a scan]
    Consider a scan of a process $p$ ending with $level[p]>0$ and $view[p]=W$ for some set $W$.
    \emph{In the context of such as scan by $p$}, we say that a processor $q$ covers a register $r$ when, after $p$ read~$r$ during its scan, $q$ writes a set of values $W'$ to $r$ such that~$W\not\subseteq W'$. %
\end{definition}

Next we state and prove an easy lemma that we will use in the proof of~\Cref{lem:durably-stored}.
\begin{lemma}
    \label{lem:easy-durably-stored}
    Consider a processor $p$ that performs a scan ending at time $t$ (when $p$ reaches label ``$write$'') with $level[p]>0$ and $view[p]=W$, for some set $W$.
    Suppose that there is a register $r$ such that, at all times during $p$'s scan, $W\subseteq r$.
    Then, at time $t$, $W$ is durably stored despite interference by $P$.
\end{lemma}
\begin{proof}
    Observe that, for every processor $q$ that has covered a register read by $p$, $q$ must read $r$ before it writes again.
    Thus, each processor $q\neq p$ can cover at most one register, and the result follows.
\end{proof}

We now turn to the proof of~\Cref{lem:durably-stored}, which we perform by induction on the natural numbers.
\begin{definition}[Induction hypothesis]
    For $1\leq i\leq N$, let $IH(i)$ assert that for every processor $p$, for every set of processors $Q$ of cardinality $i-1$ with $p\notin Q$, if
    \begin{itemize}
        \item $p$ is at label ``$write$'' with $level[p]\geq i$ and
        \item only members of $Q$ write during $p$'s last scan,
    \end{itemize}
    then $values[p]$ is durably stored despite interference by $Q\cup \{p\}$.
\end{definition}

It is easy to see that $IH(N)$ implies~\Cref{lem:durably-stored}.
We now prove by induction that $IH(i)$ holds for every $i$ with $1\leq i\leq N$.

\begin{lemma}
    $IH(1)$ holds.
\end{lemma}
\begin{proof}
    We must show that, for every processor $p$, if $p$ ends a scan with $level[p]>0$ and no processor writes during $p$'s scan, then $view[p]=W$ is durably stored despite interference by $\{p\}$ itself.
    This is trivial.
\end{proof}

\begin{lemma}
    For every $i$ such that $1\leq i<N$, $IH(i)$ implies $IH(i+1)$.
\end{lemma}
\begin{proof}
    Suppose that $1\leq i<N$ and that $IH(i)$ holds.
    Consider a processor $p$ that performs a scan ending at some time $t$ with $p$ at label ``$write$'' and $level[p]$ at least $i+1$ and $view[p]=W$ for some set $W$.
    Consider a set of processors $Q$ of cardinality $i$ and such that $p\notin Q$.
    Assume that only members of $Q$ write during $p$'s scan.
    We must show that the set of value $W$ is durably stored despite interference by $Q\cup \{p\}$.

    First, consider a processor $q\in Q$ and suppose that $p$ reads from $q$ during its scan, and suppose that $q$ started the scan leading to its write after $p$ started its scan.
    In this case, we now use the induction hypothesis to show that $W$ is durable stored despite interference by $Q\cup\{p\}$.
    Note that only the members of $Q\setminus \{q\}$ write during $q$'s scan and $|Q\setminus \{q\}| = i-1$.
    Moreover, since $p$ ends its scan with $level[p]\geq i+1$, $q$ must end its scan with $level[q]\geq i$.
    Thus, by $IH(i)$, $W$ is durably stored despite interference by $Q\setminus \{q\}$, and thus, since $W\subseteq view[q]$ and $W\subseteq view[p]$, $W$ is durably stored despite interference by $Q\cup\{p\}$.
    To sum up, the induction hypothesis implies that:
    \begin{equation}
        \parbox{.85\textwidth}{If $q\in Q$ and $p$ reads from $q$ during its scan and $q$ started the scan leading to its write after $p$ started its scan, then $W$ is durably stored despite interference by~$Q\cup\{p\}$}
    \end{equation}

    Now consider the set $R_W$ or registers $r$ such that $W\subseteq r$ when $p$ starts its scan.
    We have two cases.
    If one register in $R_W$ is never covered, then by~\Cref{lem:easy-durably-stored} we are done.

    Otherwise, note that, by Fact (1) above, we can assume that $p$ does not read from any member $q$ of $Q$ a value stemming from a scan of $q$ that started after $p$ started its scan.
    Now, to cover $R_W$, a set $Q_{\overline{W}}$ of $|R_W|$ different members of $Q$ must write, since no processor can cover two because it would need to read the not-yet-covered members of $R_W$ in between.
    Moreover, for $p$ to complete its scan, we need a set $Q_W$ of $|\overline{R_W}|$ different members of $Q$ to write a superset of $W$ in each register $r\notin R_W$, since none of them can write twice by our assumption above.
    Again by our assumption above, we must have $Q_{\overline{W}}\cap Q_W=\emptyset$, and thus $|Q|=|Q_{\overline{W}}|+|Q_W|$.
    But $|Q_{\overline{W}}|+|Q_W|=N$ and $|Q|<N$, which is a contradiction.
\end{proof}

\subsubsection{Wait-Freedom}

We now show that every processor that takes enough steps terminates.
Suppose towards a contradiction that $p$ takes infinitely many steps and does not terminate.

As we have seen, the stable-view graph of the execution is a directed acyclic graph.
Thus it must have a source $V_S$, and, after GST, all processors that have stable view $V_S$ read only from processors that have stable view $V_S$.
In this case, after GST the processors that have view $V_S$ only ever read view $V_S$ and thus, according to the algorithm, they keep incrementing their minimum level indefinitely.
Thus those processors must reach level $N$ and terminate.
This is a contradiction.

\section{Renaming}
\label{sec:renaming}

For the renaming algorithm we adapt the renaming algorithm of Bar-Noy and Dolev\cite{bar-noy_shared-memory_1989} that transforms snapshots to integers in the range $1..(N+1)N/2$.
Their algorithm proceeds as follows: Given an (non-group) implementation of the snapshot task, since we know that two snapshots of identical size $z$ must be identical, a processor can rank itself in its own snapshot, obtaining, say, ranks $r\leq z$, can now take the name $((z-1)z/2) +r$.
In other words, the algorithm reserves the integer 1 for the snapshot of size 1, integers 2 and 3 for the snapshot of size 2, the integers 4, 5, and 6 for the snapshot of size 3, etc.
A specification of the algorithm appears in~\Cref{fig:renaming}.
The algorithm is adaptive in the sense that it does not need to know $N$.

\begin{figure}[h]
\centering
\includegraphics[]{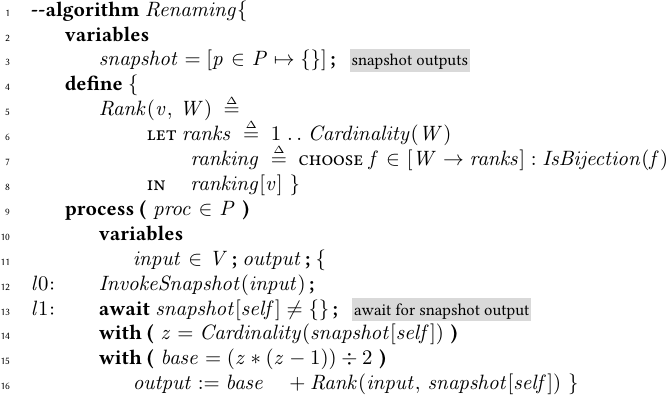}
\caption{The adaptive renaming algorithm in the PlusCal algorithm language.}
\label{fig:renaming}
\end{figure}

In our setting we must use a group implementation of the snapshot task, and we lose the fact that two snapshots of identical size are identical because processors in the same group are allowed to output incomparable snapshots.
But there is a saving grace: Notice that two incomparable snapshots must come from processors in the same group $g$, otherwise group-solvability is violated.
Moreover, any bigger snapshot returned by a different group $g'$ must include the union of the incomparable snapshots of $g$, and any smaller snapshot returned by $g'$ must be a subset of the intersection of the incomparable snapshots returned by members of $g$.

As a consequence of the above, if two processors in the same group return incomparable sets then they essentially ``reserve''  all the sizes between the intersection and the union of their snapshots; consequently they will not collide in choosing an integer with any processor from a different group.
They may nevertheless collide with each other, but this is allowed by group solvability since they are in the same group.

This subtle justification above in adapting the algorithm of Bar-Noy and Dolev is glossed over (a.k.a.\ mistake!) by Gafni~\cite{gafni_group-solvability_2004}.

\section{Obstruction-Free Consensus}
\label{sec:consensus}

To obtain an obstruction-free consensus algorithm, we first note that the snapshot algorithm of~\Cref{sec:snapshot} can easily be made long-lived, and then, similarly to Guerraoui and Ruppert~\cite{guerraoui_what_2005}, we use the long-lived snapshot in an adaptation of an algorithm of Chandra~\cite{chandra_polylog_1996}.

\paragraph{Long-lived snapshot}

In long-lived snapshot, each processor that has produced a snapshot output can invoke the snapshot again with a new input, receive a new output, invoke the snapshot again, etc.
We must guarantee that outputs only contain input values of participating processors, that the output of each processor contains all the values it has used as input so far, and that every two outputs are related by containment.

We obtain a non-blocking, long-lived snapshot algorithm by tweaking the single-shot snapshot algorithm of~\Cref{fig:snapshot}.
Processors use the algorithm of~\Cref{fig:snapshot}, keeping their local state between invocations, and, upon a new invocation, simply reset their level to $0$ and add their new input to their view. %
Since the single-use snapshot algorithm is wait-free, this long-lived snapshot algorithm is clearly non-blocking and obstruction-free.

The attentive reader might notice that this long-lived snapshot is not specified in terms of groups.
While this is inelegant, we use the definition above because it allows us to argue that the consensus algorithm below is correct for exactly the same reason as in the work of Chandra~\cite{chandra_polylog_1996}.
However, in the same vein as for tasks, we could define group solvability of long-lived problems by interpreting inputs as groups and considering that each invocation by the same processor is done by a different logical processor.
We leave it to future work to prove that the consensus algorithm below is correct if we assume it uses a group solution to long-lived snapshot.

\paragraph{Consensus}

We now put the long-lived snapshot algorithm to use by derandomizing, like Guerraoui and Ruppert~\cite{guerraoui_what_2005}, the probabilistic shared-coin algorithm of Chandra~\cite[Section 4.1]{chandra_polylog_1996}.

The algorithm appears in~\Cref{fig:consensus}, and it works as follows.
The processors each maintain a monotonically-increasing local timestamp, initialized to 0, and they repeatedly invoke the long-lived snapshot, providing as input their preferred values (which is initially their consensus input) and their timestamp.
Upon obtaining a snapshot, a processor decides a value $v$ if $v$ appears in the snapshot with a timestamp at least 2 greater than any other value; if no such value exists, the processor updates its preference to the value with the highest timestamp and its snapshot to the highest timestamp plus one.

\begin{figure}[h]
\centering
\includegraphics[]{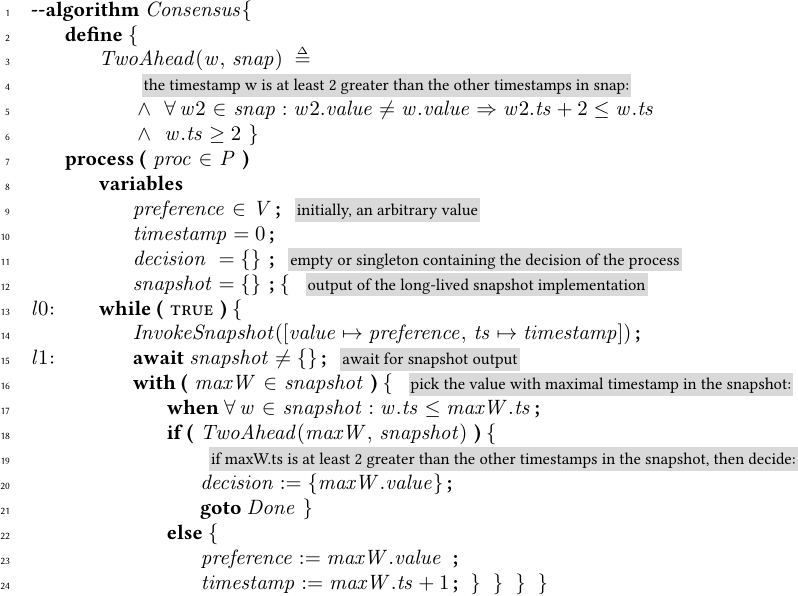}
\caption{The obstruction-free consensus algorithm in the PlusCal algorithm language. Note that the sequence of steps between any two labels (here the labels are $l0$ and $l1$) is executed atomically.}
\label{fig:consensus}
\end{figure}

Note that all communication happens through the long-lived snapshot, i.e.\ the consensus algorithm does not directly write to any register, and so there is no issue of steps of the consensus algorithm interfering with the long-lived snapshot.

The algorithm is clearly obstruction-free: if a processor runs solo it will first adopt the value with the highest timestamp and then keep increasing its timestamp and decide the value.
Safety follows like in the work of~\cite{chandra_polylog_1996}.

\section{Related work}

\paragraph{Processor anonymity}

There is a large literature studying shared-memory distributed computing with anonymous processors and non-anonymous memory, both from computability and complexity perspectives.
Below we mention only the two most relevant works.

Guerraoui and Ruppert~\cite{guerraoui_what_2005} ask what can be implemented deterministically with anonymous processors (but they do not formalize a precise notion of solvability).
They propose a wait-free atomic memory snapshot algorithm, an obstruction-free consensus algorithm, and show that an object can be implemented obstruction-free if and only if it is idempotent.
Their constructions do not work in the fully-anonymous model because the weak-counter implementation that they use in the atomic memory snapshot algorithm relies on making processors race in one direction of an array of registers to be the first to write a position.
The weak counter requires $N$ registers for a single-shot use, but it uses an unbounded number of registers in the long-lived version that they use to implement consensus.
Their implementation of the weak-counter is not possible with anonymous memory because there is no shared notion of order of the registers.

Once we obtain a snapshot implementation in the fully-anonymous, we obtain an obstruction-free consensus algorithm, like Guerraoui and Ruppert, by derandomizing the consensus algorithm of Chandra~\cite{chandra_polylog_1996}.
However, showing its correctness when substituting a group solution of the snapshot task for the atomic memory snapshots used by Guerraoui and Ruppert is subtle.

Gelashvili~\cite{gelashvili_optimal_2018} shows that any obstruction-free, processor-anonymous consensus implementation must use $\Omega(n)$ registers.
As we show in~\Cref{sec:no-less-than-n}, in the fully-anonymous model we easily obtain a lower bound of $N$ registers for any non-trivial synchronization task.

Yanagisawa~\cite{yanagisawaWaitfreeSolvabilityColorless2019} studies the wait-free solvability of colorless tasks under processor anonymity.

\paragraph{Memory anonymity}

Taubenfeld~\cite{taubenfeld_coordination_2017,taubenfeld_anonymous_2022} introduces memory-anonymous algorithms, which are algorithms for read-write shared-memory models in which the registers are anonymous, but the processors are not anonymous.
They provide algorithms, space lower bounds, and impossibility results.
A number of authors follow suit with computability and complexity results for memory-anonymous systems~\cite{aghazadeh_optimal_2019,godard_bezouts_2020,taubenfeld_set_2019}.

\paragraph{Full anonymity}
Raynal and Taubenfeld~\cite{raynal_mutual_2020} introduce the fully-anonymous model, where both processors and memory are anonymous.
They also present a mutual-exclusion algorithm relying on an atomic read-modify-write primitive.

In a paper that is the inspiration behind our work, Raynal and Taubenfeld~\cite{raynal_fully_2021} ask whether consensus and set-agreement can be solved obstruction-free in the fully-anonymous model.
They present a set-agreement algorithm for any number of processors using three registers, an algorithm for solving the consensus problem for two processors using three registers, and finally an algorithm for solving consensus for any number of processors which relies on an atomic snapshot primitive.
However, they leave the problem of implementing the atomic memory snapshot --- and thus consensus --- from individual read and writes open.
In this paper, we solve the snapshot task and the obstruction-free consensus problem, but not the atomic memory-snapshot question.
In fact, the TLC model-checker confirms that, when there are 3 processors, the algorithm of~\Cref{fig:snapshot}, which solves the snapshot task, does not provide atomic memory snapshots: in some executions, a processor returns a set of inputs $I$ such that at no point in time did the memory contain exactly the set of inputs $I$.
We conjecture that it is not possible to implement atomic memory snapshots in the fully-anonymous model.

Finally, Imbs et al.\cite{imbs_election_2022} study election problems in variants of the fully-anonymous model with read-modify-write primitives.

\paragraph{Register allocation}

Results unrelated to anonymity but which also encounter the difficulty of processors write-stepping on each other appear in~\cite{delporte-gallet_adaptive_2013} and in~\cite{delporte-gallet_linear_2015}.

\paragraph{Group solvability}

Gafni~\cite{gafni_group-solvability_2004} introduces the notion of group-solvability of tasks and studies group solvability with groups of infinite size in the standard, non-anonymous model.
His work adapts the algorithm in~\cite{bar-noy_shared-memory_1989} for a wait-free algorithm that group-solves renaming with $N(N+1)/2$ names, which we reuse in this paper.
As we note in~\Cref{sec:renaming}, the correctness of the algorithm when using a group solution to the snapshot task instead of atomic memory snapshot is non trivial and was missed in Gafni's paper.
Gafni also shows that immediate-snapshot for $3$ processors is not wait-free group solvable.

\section{Conclusion}

In this paper we approach anonymity and full-anonymity not by consider ad-hoc synchronization problems but in general by starting with tasks.
We connect processor anonymity with group solvability and we ask the eventual-pattern question, which uncovers a key structural property of read-write computation in the fully-anonymous model.

Our results have implications for other models.
For instance, if we consider just processor anonymity, it is not clear that the results of Guerraoui and Ruppert~\cite{guerraoui_what_2005} can be achieved with just $N$ registers.
The results of this paper, obtained under full anonymity, clearly show that the answer is yes.
Thus in terms of the possibility of solving a task within some minimal number of registers this paper resolves the question in the weakest model.

On the negative side, the connection with solvability in terms of groups and the results of Gafni\cite{gafni_group-solvability_2004} imply that the immediate-snapshot task is not group-solvable even if just processors, and not memory, are anonymous.
Thus it is also not solvable in the fully-anonymous model.
Stepping outside the realm of tasks, we also conjecture that the atomic memory snapshot problem is not solvable in the fully-anonymous model.

In the one anonymity combination not investigated, where only memory is anonymous, the group connection is lost.
Hence it is not a-priori clear that immediate-snapshots is not solvable with just memory anonymity.

Finally, we conjecture that, under wait-freedom, a model with just memory anonymity with $N$ registers is as powerful as the single-writer multi-reader read-write, fully non-anonymous model.
Last but not least, in subsequent work we similarly intend to prove that adding memory anonymity to processor anonymity is no real hindrance --- whatever task is solvable wait-free with just processor anonymity is solvable also when we add memory anonymity.

\bibliographystyle{ACM-Reference-Format}
\bibliography{references}

\end{document}